\documentclass[draftcls, 12pt,onecolumn,oneside]{IEEEtran}
\usepackage{amsfonts}
\usepackage{amssymb}
\usepackage{amsmath}
\usepackage{mathrsfs}
\usepackage{verbatim}
\usepackage{subfigure}
\usepackage{balance}
\usepackage{booktabs}
\usepackage{fancybox}
\usepackage{bm}
\usepackage{extarrows}
\usepackage{algorithm}
\usepackage{algorithmic}
\usepackage{multirow}
\usepackage{array}
\usepackage{graphicx}
\usepackage{epstopdf}
\usepackage{cite}
\usepackage{subfig}

\newtheorem{proof}{Proof}
\newtheorem{proposition}{Proposition}

\begin{document}
\title{Secrecy and Covert Communications against UAV Surveillance via Multi-Hop  Networks}
\author{  Hui-Ming Wang, \IEEEmembership{Senior Member, IEEE,} Yan Zhang, Xu Zhang, \\and Zhetao Li, \IEEEmembership{Member, IEEE}
\thanks{H.-M. Wang, Y. Zhang and X. Zhang
are with the School of Electronic and Information Engineering, Xi'an Jiaotong University, and also with the Ministry of Education Key Laboratory for Intelligent Networks and Network Security, Xi'an 710049, China (e-mail: xjbswhm@gmail.com; yzhangxjtu@163.com; jcx8008208820@stu.xjtu.edu.cn).}
\thanks{Z. Li is with Key Laboratory of Hunan Province for Internet of Things and Information Security, and also with The
College of Information Engineering, Xiangtan University, Hunan, China (email: liztchina@gmail.com)}
}

\maketitle

\begin{abstract}
The deployment of unmanned aerial vehicle (UAV) for surveillance and monitoring gives rise to the confidential information leakage challenge in both
civilian and military environments.
The security and covert communication problems for a pair of terrestrial nodes
against UAV surveillance are considered in this paper.
To overcome the information leakage and increase the transmission reliability,
a multi-hop relaying strategy is deployed.
We aim to optimize the throughput by carefully designing the parameters of the multi-hop network, including the coding rates, transmit power, and required number of hops.
In the secure transmission scenario, the expressions of the connection probability and secrecy outage probability of an end-to-end path are derived and the closed-form expressions of the optimal transmit power, transmission and secrecy rates under a fixed number of hops are obtained.
In the covert communication problem, 
under the constraints of the detection error rate and aggregate power, the sub-problem of transmit power allocation is a convex problem and can be solved numerically.
Simulation shows the impact of network settings on the transmission performance.
The trade-off between secrecy/covertness and efficiency of the multi-hop transmission is discussed which leads to the existence of the optimal number of hops.

\end{abstract}

\section{Introduction}
The potentially wide applications of unmanned aerial vehicles (UAVs), both in military and civilian fields, have given rise to a significant research enthusiasm in the academic society recently. Due to the advantages such as high mobility and low cost, UAVs have been widely applied as platforms in security monitoring, surveillance and reconnaissance missions,
environmental inspections and so on. UAVs have also been highly promising for playing an important role in future wireless communication systems  \cite{survey_RZHANG}. 
Taking advantage of line-of-sight (LoS) connections established in the UAV-assisted communication systems, the coverage and quality-of-service (QoS) of the networks can be enhanced efficiently.
Owing to their controllable mobility and
cost-effective deployment,
UAVs can work as aerial base stations to provide offloading services in extremely crowded areas and to enable communications when terrestrial networks are damaged due to  sudden natural disasters \cite{uav_trans_receive1,uav_trans_receive2,uav_broadcast}.
In addition, they can be used as mobile relays to establish wireless connectivity for distant users that cannot communicate with each other by reliable direct links \cite{uav_relay1,uav_relay3}.


To facilitate an effective UAV communication and to evaluate its performance, it is crucial to characterize the air-to-ground wireless channel accurately and establish reliable analytical channel models.
In \cite{channel_freespace1,channel_freespace2}, the authors assume that the air-to-ground channel is dominated by a LoS link and therefore follows the free-space pass loss model. To further take into account the small-scale fading, in \cite{channel_small1,channel_small2}, the channel is considered following the Rayleigh fading model. While in \cite{nakagami} and \cite{rician}, the terrestrial-aerial channels are considered as a combination of LoS and scattered components, and are modeled as Nakagami-$n$ and Rician fading models, respectively.
In \cite{channel_angel1} and \cite{opt_lap_altitude}, 
a hybrid model with the LoS and non-line-of-sight (NLoS) components as well as the corresponding occurrence probabilities is considered.
The probability of LoS is related to the environment and the elevation angle of the communication link. We can see that due to the large elevation angle, the LoS component of the air-to-ground channel is dominant and the channel quality is generally much better than the terrestrial fading channel, where the blockage, shadowing effect and various scatters deteriorate the channel significantly.

\subsection{Motivation}
However, the UAVs bring in challenges to information security society. Due to the openness of wireless media and the broadcast nature of wireless communications, any receiver located in the cover range of the transmitter can receive the transmitted signal. UAVs can be deployed as an eavesdropper, which gives rise to the confidential information leakage of wireless transmissions. Especially, benefiting from the advantages of the air-to-ground channel properties compared with those of the terrestrial channel \cite{channel_freespace1} -\cite{Liu_uav_active_eve}, UAVs are highly likely to intercept or detect the transmitted signals when the source node communicates with the receiver directly under severe terrestrial fading channels. One recent work \cite{PELE} considers the scenario in which the UAV acts as a surveillance to eavesdrop the transmissions between the suspicious transceivers on the ground.
The eavesdropping performance is improved by jamming the suspicious communication.
In \cite{uav_eave}
a multi-antenna UAV full-duplex eavesdropper wiretaps and jams the communication between the transmiter and receiver simultaneously by adjusting its transmit and receive beamforming vectors.
In order to confront such threatens, it is of great importance to explore methods protecting the legitimate transmission and preventing information from being intercepted or detected by the hostile UAVs.

To deal with the secure transmission problem, the traditional methods are mainly based on encryption technologies and other higher layer security protocols. More recently, physical layer security (PLS) has emerged as a complementary and promising security solution to guarantee communication secrecy \cite{hwang_phsuav}. Utilizing the physical characteristics of the wireless channels, PLS can achieve transmission secrecy in the sense of information theoretical security \cite{wyner}. This solution has been widely applied in various wireless networks, such as 5G networks \cite{Yang2015Safeguading}, heterogeneous networks \cite{Wang2016Physical}, \cite{Wang2016Physical_book}, cooperative relay networks \cite{Wang2015Enhancing}, etc. Recently,
a number of studies have applied PLS to improve the security of the UAV communication systems in which the UAV works as part of the legitimate communication systems. For instance,  \cite{Liu_uav_active_eve} and \cite{channel_freespace2} investigate the secure communication with UAV receivers in the presence of passive and active eavesdroppers, respectively. 
\cite{R1_review1} and \cite{R2_review1} studied the trajectory and power control problems for UAV-ground systems under one potential ground eavesdroppers and multiple eavesdroppers with imperfect knowledge of locations, respectively.
\cite{R3_review1} investigated another scenario where one legitimate monitor tries to improve its eavesdropping capability toward a UAV-aided system via proactive jamming.

PLS is in the sense that the eavesdropper can receive the signal but can not demodulate it to obtain the correct information bits.
To further improve the security, the transmitter may wish to send  messages through the wireless links without being detected by potential eavesdroppers. In other word, the eavesdropper may not realize the existence of  a wireless communication activity. This is also called as \emph{low probability of detection (LPD)}, or, \emph{covert communication} \cite{limit_awgn}.
The basic theory and the performance limit of the covert communication in additive white Gaussian noise (AWGN) channel are discussed in \cite{limit_awgn} which indicates that at most $\mathcal{O}(\sqrt{n})$ bits can be transmitted to the legitimate receiver reliably in $n$ channel uses, without being detected by the warden.
The asymptotic behavior has also been investigated in other kinds of channels such as the discrete memoryless channel and the binary symmetric channel \cite{memoryless}, \cite{binary}.
Based on these pioneering works, the covert communication has been studied under many circumstances, such as noise uncertainty \cite{noise_uncertain}, channel information uncertainty \cite{channel_uncertain}, finite blocklength \cite{finite_length}, and multi-hop routing network \cite{covert_multihop}.

From the previous discussion, we have known that 
the air-to-ground channels have better characteristics and are modeled differently comparing with terrestrial links, which would increase the stress of communication networks on information confidentiality, as well as result in different considerations conclusions when considering network designing problems in such circumstances.  
With the increasingly wider application of UAV in communication and its potential threaten toward information security, 
secure and covert communications under aerial surveillance  require in-depth studies.
However, only a few research works consider this scenario. 
To the best of our knowledge, only \cite{PELE} and \cite{uav_eave} discussed the PLS transmission with UAV wiretappers. Yet the two works studied ways to improve the eavesdropping performance, not to confront UAV's wiretapping 
and guarantee secure transmission from the terrestrial network's perspective. 
In addition, no previous work has investigated the LPD communication problem under the surveillance of the UAV.

\subsection{Contributions of Our Work}
In this paper, we evaluate and optimize the secrecy and covert communications, respectively,  between a transmitter and a receiver on the ground surveilled by a UAV.
As previously mentioned, the advantages of the air-to-ground channels compared with the terrestrial channels
 lead to the proneness of legitimate information leakage and detection.  
Specially, \emph{the transmitted signal combating the severe fading of the terrestrial channel to be received reliably at the destination is more inclined to be wiretapped by the UAV}.
To solve this problem, one helpful solution
is adopting multi-hop strategy, i.e., multiple intermediate nodes relay the signals from the source to the destination. On one hand, the multi-hop strategy will strengthen the connection of terrestrial network. On the other hand, the transmit power of each node would be reduced sharply compared to the direct transmission from the source to destination, which could drastically decrease the probability of being discovered and detected by the hostile UAV.
This strategy is especially suitable for networks with limited energy and large distance between the source-destination pair, compared with other ways sending signals directly to the destination, e.g. the cooperative remote jamming method in \cite{R1_review3}.

We need to address that
there naturally exists \emph{a non-trivial trade-off between the secrecy/covertness and the efficiency of the multi-hop transmission}, which is highly related to
the  number of hops and the transmit power of each node.
With fewer hops, the transmit power allocated to each node should be sufficiently large in order to guarantee the QoS and avoid connection outage in each hop, which increases the risk of being wiretapped.
On the other hand, the transmit power could be reduced with more hops to avoid the information leakage, but the effective throughput of the network would be influenced which may lead to a lower efficiency. In addition, too many hops will also increase the risk of being detected.
Therefore, the design and optimization of system parameters for the security and covert transmissions using multi-hop relaying is the main focus of this paper.

In our work, the randomness with regard to the LoS and NLoS characteristics is considered when modeling the air-to-ground channels. The channel between one ground device and the UAV can either be LoS or NLoS, and is modeled with different channel coefficients under LoS and NLoS conditions, which is a significant difference from terrestrial wiretapping circumstances.
Under the proposed multi-hop transmission strategy, our objective is to maximize the throughput by optimizing the network parameters including the coding rates, transmit power, and required number of hops.
In the secure transmission problem, the secrecy outage probability must be kept under a tolerant level.
While in the covert communication problem, the detection error rate at the UAV is required to be higher than a predetermined threshold.
The main contributions are summarized as follows:

1) For the secrecy transmission problem, the exact expressions of the connection probability and the secrecy outage probability of an end-to-end path are derived. 
The secrecy throughput maximization problem under the constraints of the secrecy outage probability as well as the aggregate transmit power is discussed. 
The closed-form expressions of the optimal transmit power, transmission and secrecy rates are obtained. Further more, the optimal trade-off with respect to the transmission security and efficiency is evaluated.

2) For the covert communication problem, the detection error rate is used to measure the covertness performance, and a lower boundary is constructed using the relative entropy between the two joint distributions under different hypotheses. 
The optimization of transmit power at each hop turns to be a convex optimization problem and could be solved numerically. 
The expression of the optimal transmission rate is also derived, and the optimal number of hops is evaluated.

Note that though there exist some differences between the secrecy and covert communications in problem construction and solving procedures, 
the considered problems in secrecy and covert communication parts are closely related. 
First, both problems aim to guarantee the transmission security in the multi-hop networks. Secrecy communication prevents information being demodulated and covert communication conceals the transmission from being detected by the UAV.
Second, there exists a trade-off between the security and efficiency in both secrecy and covert communications. The increase of number of hops will degrade UAV's wiretapping and detecting ability and result in the increase of transmission rate, but more time for receiving information will be required and more chances for wiretapping/detecting signals will be provided. 
In addition, both problems maximize the end-to-end throughput, and the 
opposite influence of transmission rate on connection probability and the coding rate which are two important component of the throughput 
leads to the existence of optimal secrecy rate in secrecy communication and transmissions rate in covert communication.


\subsection{Related Works}
More recently, PLS and covert communications have also been extended to multi-hop relaying systems \cite{covert_multihop,adhoc1,adhoc12,linear}.
Specifically, in \cite{covert_multihop}, the optimal path is chosen for a covert communication network in the presence of multiple collaborating wardens. \cite{adhoc1} focuses on the secure routing and power allocation problem with the deployment of cooperative jamming and the knowledge of the eavesdroppers, while \cite{adhoc12} deals with the routing and power optimization under randomly distributed eavesdroppers and one multi-antenna jammer. In \cite{linear}, the design and secrecy performance of a linear multi-hop network with randomly distributed eavesdroppers are analyzed under on-off and non-on-off transmission schemes respectively.
 The aforementioned works all consider the terrestrial eavesdropping/detecting circumstances, which quite differs from our aerial surveillance scenario and have different research targets from ours.

The remainder of the paper is organized as follows: In Section II, the system model of the multi-hop relay network is presented, and the secrecy transmission optimization problem as well as the covert communication problem are formulated. In Section III, the parameters of the network are optimized sequentially. The solution of the covert communication problem is provided in Section IV.
Section V presents the numerical results and Section VI draws the conclusions.

\section{System Model}

\begin{figure}
	\centering
	\includegraphics[width=3.5in]{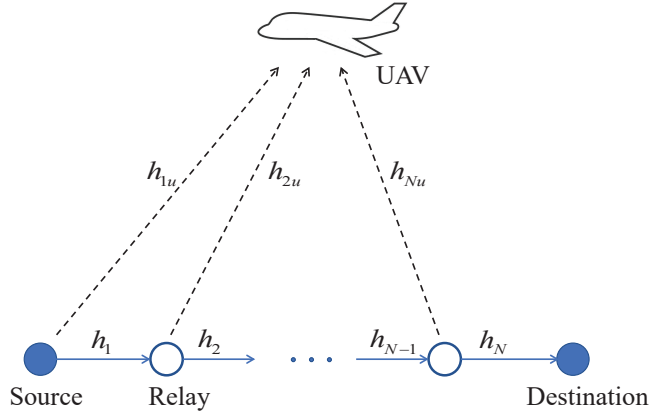}
	\caption{The system model with linear multi-hop relaying network eavesdropped by UAV.}
	\label{system_model}
\end{figure}

We consider the communication between one pair of source and destination nodes with distance $D$. One UAV in the air monitors any suspicious communication and wiretaps the transmitted signal. To avoid being detected by the UAV and to guarantee the security transmission, the multi-hop relaying is deployed and the decode-and-forward (DF) protocol is considered for confidential signal convey.
The source, destination, relays, as well as the UAV are all equipped with a single omnidirectional antenna, respectively.
We assume that the UAV has a fixed location in the air which is known by the ground network. 
The relays are placed equidistantly on the line from the source to the destination, so as to simplify the investigation toward the influence of the number of hops \cite{linear1}. 
Supposing that there are $N-1$ relays in the network, then the signals will be transmitted via $N$ hops through the end-to-end route from the source to the destination, and the distance between the transmitter and receiver of each hop can be denoted as $d_{tr}=\frac{L}{N}$.
The system model of this linear relaying network is illustrated in Fig. \ref{system_model}.

The wireless channels between the terrestrial nodes are subject to a small-scale Rayleigh fading together with a large-scale path loss.
The path loss exponent is $\alpha$.
The Rayleigh fading coefficient for the $n$-th hop is denoted as $h_n$, which is modeled as an independent complex Gaussian variable with zero mean and unit variance, i.e. $h_n\sim\mathcal{CN}(0,1)$.
For the air-to-ground channels, we adopt the hybrid model proposed in \cite{channel_angel1}.
Specifically, we suppose that the probability for each path being LoS is $P_{los}$,
where the value of $P_{los}$ primarily depends on the propagation environment and the elevation angle \cite{Liu_uav_active_eve}. The probability for the air-to-ground channel from the $n$-th terrestrial transmitter node to the UAV being LoS can be expressed as
\begin{equation}
P_{los}(\theta_n)=\frac{1}{1+C\exp(-B(\theta_n-C))},
\end{equation}
where $B$ and $C$ are constant depending on the environment and $\theta_n$ is the elevation angle which can be further written as
\begin{equation}
\theta_n=\frac{180}{\pi}\sin^{-1}\left(\frac{H}{d_{nu}}\right),
\end{equation}
with $H$ representing the height of the UAV, and $d_{nu}$ the distance between the $n$-th transmitter and the UAV.
The probability for the path being NLoS is
$P_{nlos}(\theta_n)=1-P_{los}(\theta_n)$.

The small-scale fading coefficient 
from the transmitter of the $n$-th hop to the UAV, denoted as $h_{nu}$, follows Rayleigh fading. For large-scale path loss, LoS and NLoS links possess different fading coefficients $a_{nu}$, which can be expressed as:
\begin{equation}
a_{nu}=\left\{
             \begin{array}{lcl}
           \sqrt{\lambda_0}  d_{nu}^{-\beta_1/2} &\quad \mathrm{LoS\  link} \\
             \sqrt{\eta \lambda_0} d_{nu}^{-\beta_2/2} &\quad\mathrm{NLoS\ link}
             \end{array}
        \right.
        ,
        \end{equation}
where $\beta_1$ and $\beta_2$ represent the path loss exponents for the LoS link and NLoS link respectively,
$\lambda_0$ is the path loss at the reference distance $d_0=1$m,
and $\eta$ is the excessive path-loss coefficient, which shows an additional attenuation due to the 
higher path loss of the NLoS connection.

Under the DF relaying protocol, at the $n$-th hop, the received signals of the $n$-th terrestrial receiver and the UAV $y_{R_n}$ and $y_{E_n}$ are written as
\begin{align}
&y_{R_n} = \frac{\sqrt {p_{n}\lambda_0} h_n}{d_{tr}^{\alpha /2}}{x_{n}} + n_{R_n}\label{y_n} \  \  \text{and} \\
&y_{{E_n}} =\sqrt{p_n} a_{nu} h_{nu} {x_n} + n_{E_n}, \label{y_e}
\end{align}
respectively. $x_n$ and $p_n$ are the transmitted symbol and corresponding transmit power at the $n$-th hop, respectively,  and $n_{R_n}$ and $n_{E_n}$ represent the noise at the $n$-th terrestrial receiver and the UAV following $\mathcal{CN}(0,\sigma_0^2)$.
Defining $\sigma^2\triangleq \frac{\sigma_0^2}{\lambda_0}$, the expressions of the corresponding SNR can then be expressed as
 \begin{align}
 &\mathrm{SNR}_{R_n}=\frac{p_{n}|h_n|^2}{d_{tr}^\alpha\sigma^2}\label{SNR_R} \  \  \text{and}\\
&\mathrm{SNR}_{E_n}=\frac{p_{n}a_{nu}^2|h_{nu}|^2}{\lambda_0\sigma^2}. \label{SNR_E}
 \end{align}

With the expressions of the received signals and SNRs as well as other notations above, in the following, we formulate optimization problems in secrecy and covert transmissions, respectively. Due to the differences between the secrecy and covert communications, different signaling models, knowledge of the channel state information and optimization targets are assumed in these two conditions. 

\subsection{Secrecy Transmission}
We first consider the PLS transmission strategy under the eavesdropping of the UAV.
Under Wyner's encoding scheme, different codebooks are adopted to retransmit the signal in multiple DF hops, as proposed in \cite{DFsecure}. Hence the UAV 
are not able to jointly decode the signals received from multiple hops, but can only decode these signals individually. We denote the transmission rate and secrecy rate in $n$-th hop as $R_{t,n}$ and $R_{s,n}$, respectively.

For the $n$-th receiver on the end-to-end link,
the legitimate node can decode the message correctly if its channel capacity $C_{t,n}$ is greater than the transmission rate $R_{t,n}$, i.e. $C_{t,n}>R_{t,n}$, which is equivalent to $\mathrm{SNR}_{R_n}>\gamma_{c,n}\triangleq 2^{R_{t,n}}-1$. According to Wyner's wiretap coding theory, the transmission confidentiality fails when the capacity of the channel from the $n$-th transmitter toward the UAV $C_{e,n}$
is greater than the rate redundancy $R_{e,n}\triangleq R_{t,n}-R_{s,n}$, i.e., $C_{e,n}>R_{e,n}$,
which is equivalent to $\mathrm{SNR}_{E_n}>\gamma_{e,n}\triangleq 2^{R_{e,n}}-1 $. $R_{e,n}$ reflects the ability to secure the transmission against the wiretapping.

To convey the end-to-end confidential information under the multi-hop DF protocol and guarantee the consistency of the message, the secrecy rates for all hops should be identical, i.e,. $R_s=R_{s,1}=\cdots=R_{s,N}$, while $R_{t,n}$ (and hence $R_{e,n}$) can be different for different hops \cite{DFsecure}. 
Considering the complexity of the optimization problem when separate transmission rates are adopted, 
in this paper we suppose that fixed-rate $R_t$ is used to simplify the problem. Hence we have $R_t=R_{t,1}=...R_{t,N}$.
Based on this assumption and the independence of the received signals for different hops, we now derive the definitions of the connection and secrecy outage probabilities
for an end-to-end link.
The connection probability is defined as the probability that each receiver can decode the message correctly, which can be written as:
\begin{equation}\label{def_pc}
P_c\triangleq\prod_n \mathrm{Pr}\left(\mathrm{SNR}_{R_n}>\gamma_c\right),
\end{equation}
where $\gamma_c=2^{R_t}-1$.
The secrecy transmission occurs when $\mathrm{SNR}_{E_n}$ for each hop is less than $\gamma_e=2^{R_e}-1$. Therefore, the secrecy outage probability can be defined as
\begin{equation}\label{def_pso}
P_{so}\triangleq1-\prod_n \mathrm{Pr}\left(\mathrm{SNR}_{E_n}<\gamma_e\right).
\end{equation}

The end-to-end \emph{secrecy throughput} reflecting the spectral efficiency of a complete signal transmission from the source to the destination \cite{DFsecure}, which is defined as
\begin{equation}\label{def_throughput}
\Phi_S\triangleq \frac{1}{N}P_c R_s.
\end{equation}
Our objective is to maximize the secrecy throughput for the multi-hop network under the secrecy outage constraint as well as the sum power constraint,  by optimizing the number of hops required, the transmission rate, the secrecy rate and the transmit power for each node. The problem can be formulated as
\begin{subequations}\label{secure_prob}
\begin{align}
\max_{N,R_t,R_s,p_n}&\ \frac{1}{N}P_c R_s \label{secure_prob_obj}\\
\mathrm{s.t.}&\ P_{so}\leq\zeta \label{secure_prob_cstr_pso}\\
&\ \sum_n p_n \leq P_T,
\end{align}
\end{subequations}
in which $\zeta$ represents the constraint of $P_{so}$, and $P_T$ is the maximal value of sum power.
This problem will be solved in detail in Sec. III.

\subsection{Covert Transmission}
In this subsection, we consider the covert communication strategy for the multi-hop network.
The UAV has the ability to detect whether there exists signals transmitted. In order to avoid being detected by the UAV, the network on the ground needs to design the transmission carefully.

In order to detect the existence of the signal transmission, the UAV is required to distinguish between the
following two hypotheses of the received signals:
\begin{equation}
y_n(l)=\left\{
\begin{array}{lcl}
            \sqrt{p_n} a_n h_{nu} x_n(l)+n_{E_{n}}(l) &\quad H_1\\
              n_{E_{n}}(l) &\quad  H_0
             \end{array}
        \right.
         l=1,...,L.
\end{equation}
where $L$ represents the length of the codeword in a channel coherent block. $H_1$ denotes the hypothesis that there is signal transmission on the ground, and $H_0$ means no  transmission happening.
To facilitate the analysis, $x_n(l)$ is supposed to be a phase shift keying (PSK) signal with unit amplitude and its phase denoted as $\phi$. Then under the hypothesis $H_1$, we define
$h'_{nu}\triangleq h_{nu} e^{j\phi}$, and $h'_{nu}$ still follows $\mathcal{CN}(0,1)$.
Therefore, the distribution of the received signal at the UAV follows
\begin{equation}
\left\{
\begin{array}{lcl}
              \mathcal{CN}(0,p_n a_n^2+\sigma_0^2)&\quad H_1\\
               \mathcal{CN}(0,\sigma_0^2)&\quad  H_0
             \end{array}
        \right.        
\end{equation}


As a detector, the UAV aims to make more correct decisions and maintain low possibility of detection error which include the probability of missed detection and the probability of false alarm $P_{FA}$. The probability of missed detection $P_{MD}$ represents the possibility that the UAV declares nonexistence of communication while the terrestrial nodes transmit signals, i.e. $P_{MD}=P\left\{H_0|H_1\ \mathrm{is\ true}\right\}$. While  $P_{FA}$ denotes the condition when UAV declares the existence of communication yet no signal transmission happens, i.e. $P_{FA}=P\left\{H_1|H_0\ \mathrm{is\ true}\right\}$. 

We assume that the UAV makes a decision based on all signals it received on the multi-hop transmission. To guarantee the covertness, the probability of detection error must satisfy
\begin{equation} \label{def_detection_err}
P_{MD}+P_{FA}\geq 1-\epsilon,
\end{equation}
where $\epsilon$ can reflect the strictness of the constraint on the detection error rate and usually has a small value. Assuming that the UAV performs the optimal test, then according to the Pinsker's inequality \cite{limit_awgn}, we can obtain a lower boundary of the probability of detection error expressed as
\begin{equation} \label{upbound_detection_err}
P_{MD}+P_{FA}\geq 1-\sqrt{\frac{1}{2}D(Q_1||Q_0)},
\end{equation}
where $Q_1$ ($Q_0$) represents the joint probability distribution of all the UAV's received signals when hypothesis $H_1$ ($H_0$) is true. $D(Q_1||Q_0)$ is the relative entropy between $Q_1$ and $Q_0$, and is defined as 
\begin{equation}
D(Q_1||Q_0)\triangleq\int Q_1(x)\ln\frac{Q_1(x)}{Q_0(x)}dx.
\end{equation}
Considering (\ref{def_detection_err}) and (\ref{upbound_detection_err}) jointly, we can obtain 
\begin{equation} \label{entropy_constraint}
D(Q_1||Q_0)\leq 2\epsilon ^2.
\end{equation}
When (\ref{entropy_constraint}) is satisfied, it is  undoubtedly that  (\ref{def_detection_err}) is also satisfied.  Therefore, (\ref{entropy_constraint}) is a tighter constraint on the detection error rate $P_{MD}+P_{FA}$. 

Similar to the secrecy transmission condition, different code books are adopted to re-transmit the signals in multiple DF hops. Hence the received signals over different hops are independent with each other, and we have
\begin{equation}
Q_i=\prod_{n=1}^N Q_i^{(n)},\ i=0,1 ,
\end{equation}
in which $Q_i^{(n)}$ is the joint probability distribution of the UAV's received signals over the $n$-th hop under hypothesis $H_i$.  
Due to the definition of relative entropy and the independence of received signals over different hops, 
\begin{align}
D(Q_1||Q_0)=\sum_n D(Q_1^{(n)}||Q_0^{(n)}).
\end{align}

In covert communication situation, the terrestrial network cares little about confronting UAV's decoding and the secrecy rate, which differs from the secrecy communication.
We consider \textit{transmit throughput} $\Phi_C\triangleq P_c R_t/N $ to reflects the transmission efficiency.
The optimization objective is to maximize the end-to-end transmit throughput $\Phi_C$ under constraints on the relative entropy $D(Q_1||Q_0)$ and sum of transmit power over different hops. 
The optimization problem can be formulated as
\begin{subequations}\label{detect_problm}
\begin{align}
\max_{N,R_t,p_n}&\ \frac{1}{N}P_c R_t\\
\mathrm{s.t.}&\ D(Q_1||Q_0)\leq 2\epsilon ^2,\\
& \sum_n p_n\leq P_T.
\end{align}
\end{subequations}


\emph{Remark}: The objectives in
(\ref{secure_prob}) and (\ref{detect_problm}) can reflect the consideration on reliability and delay in multi-hop transmissions. 
$P_C$ represents the probability that the destination can decode the message correctly, which can reflect the network's reliability. The number of hops $N$ will influence the delay of one source-destination transmission. In secrecy/covert transmission problems, both $P_C$ and $N$ would affect the security/covertness of the multi-hop network. To balance the  security, reliability and delay of the network, $\Phi_S$ and $\Phi_C$ are considered to optimize the performance of multi-hop transmission networks.

Till now we have obtained the secrecy and covert communication problems in (\ref{secure_prob}) and (\ref{detect_problm}).
In the following two sections, 
we will provide comprehensive analysis and solution methods of the
two problems. 

\section{Security Transmission Optimization}
In this section, we aim to solve the security transmission problem (\ref{secure_prob}). We first provide the closed-form expressions of the connection probability and the secrecy outage probability.
According to the definition in (\ref{def_pc}), the connection probability for any given $N$ is given by
 \begin{align}\label{expr_pc}
P_c&=\prod_n\mathrm{Pr}\left\{\frac{p_{n}|h_n|^2}{d_{tr}^\alpha\sigma^2}>\gamma_c\right\}\nonumber\\
&=\prod_n\mathrm{Pr}\left\{|h_n|^2>\frac{\gamma_c d^\alpha_{tr}\sigma^2}{p_n}\right\}\nonumber\\
&=\exp\left(-\sum_n\frac{\gamma_c d^\alpha_{tr}\sigma^2}{p_n}\right),
 \end{align}
where (\ref{expr_pc}) is due to the exponential distribution of the fading coefficient $|h_n|^2$.

Before obtaining the detailed expression of secrecy outage probability $P_{so}$ for the end-to-end transmission, we first consider the security probability at the $n$-th hop $P_s^n$, which can be defined as
\begin{align}\label{expr_psn}
P_{\mathrm{s}}^n&\triangleq\mathrm{Pr}(\mathrm{SNR}_{En}<\gamma_e)\nonumber\\
&\overset{(a)}{=}P_{los}(\theta_n)\mathrm{Pr}(\mathrm{SNR}_{En}<\gamma_e|\mathrm{LoS}) +\left(1-P_{los}(\theta_n)\right)\mathrm{Pr}(\mathrm{SNR}_{En}<\gamma_e|\mathrm{NLoS})\nonumber\\
&=P_{los}(\theta_n)\mathrm{Pr}\left(|h_{nu}|^2<\frac{\gamma_e d_{nu}^{\beta_1}\sigma^2}{p_n}\right)+\left(1-P_{los}(\theta_n)\right)\mathrm{Pr}\left(|h_{nu}|^2<\frac{\gamma_e d_{nu}^{\beta_2}\sigma^2}{\eta p_n}\right)\nonumber\\
&\overset{(b)}{\approx}\mathrm{Pr}\left(|h_{nu}|^2<P_{los}(\theta_n)\frac{\gamma_e d_{nu}^{\beta_1}\sigma^2}{p_n}+\left(1-P_{los}(\theta_n)\right)\frac{\gamma_e d_{nu}^{\beta_2}\sigma^2}{\eta p_n}\right)\nonumber\\
&=1-\exp\left(-\frac{\gamma_e \sigma^2}{\eta p_n}\left(P_{los}(\theta_n)\eta d_{nu}^{\beta_1}+\left(1-P_{los}(\theta_n)\right) d_{nu}^{\beta_2}\right)\right),
\end{align}
where $(a)$ holds for the probability distribution of the LoS and NLoS air-to-ground channels, and
 $(b)$ is due to the Jensen's inequality.
Substituting  (\ref{expr_psn}) into (\ref{def_pso}), the secrecy outage probability of the entire path is
\begin{equation}\label{expr_pso}
\begin{split}
P_{\mathrm{so}}&=1-\prod_{n=1}^N P_{\mathrm{s}}^n\\
&=1-\prod_{n=1}^N\left\{1-\exp\left(-\frac{\gamma_e \sigma^2}{\eta p_n}\left(P_{los}(\theta_n)\eta d_{nu}^{\beta_1}+\left(1-P_{los}(\theta_n)\right) d_{nu}^{\beta_2}\right)\right)\right\}.
\end{split}
\end{equation}
(\ref{expr_pc}) and (\ref{expr_pso}) indicate that the increase of transmit power can lead to the raise of both probabilities of connection and information leakage. Therefore, the transmit power is required to be designed carefully in order to optimize the secrecy performance of the multi-hop network.

From (\ref{expr_pc}) and (\ref{expr_pso}), we know that the connection probability $P_c$ is a function of $N$, $R_t$ and $p_n$.
Meanwhile, $R_t$ and $p_n$ are coupled with $R_s$ due to the constraint of the secrecy outage probability shown in (\ref{secure_prob_cstr_pso}). Hence $P_c$ can be expressed as the function of all variables $N, p_n, R_t$, and $R_s$.
Therefore, problem (\ref{secure_prob_obj}) can be solved by optimizing $N$, $R_s$, $R_t$, and $p_n$ sequentially, i.e.
\begin{equation}\label{secure_steps}
\max_{N,R_t,R_s,p_n} \frac{1}{N}P_c R_s \Leftrightarrow\max_N\frac{1}{N}\left(\max_{R_s}R_s\left(\max_{R_t,p_n}P_c\right)\right).
\end{equation}
Based on the expressions of the connection probability and  secrecy outage probability of the multi-hop network,
we propose a method to solve this problem and provide a globally optimal solution\footnote{Note that though some approximations have been applied, the solution is not locally optimal.}. 
The method can be divided into three parts: First, we optimize the transmit power $p_n^*$ and $R_t^*$ for each node under any given $R_s$ and $N$; Then the optimal
$R_s^*$ under any $N$ is optimized further; Finally, the number of hops optimizing the throughput is found. In each part, the specific expressions of the optimal transmit powers, transmission rate, and secrecy rate are derived, and the optimal number of hops can be obtained using a one-dimensional searching method.

\subsection{Transmit Power and Rate Optimization}
In this part, we focus on optimizing the transmit power $p_n$ and transmission rate with fixed values of $R_t$ and $N$ under the secrecy outage constraint.
This sub-problem can be written as the following form:
\begin{subequations}
	\begin{align}
 \max_{p_n, R_t}&\ \exp\left(-\sum_n\frac{\gamma_c d^\alpha_{tr}\sigma^2}{p_n}\right)\\
 \mathrm{s.t.}&\ 1-\prod_n\left\{1-\exp\left(-\frac{\gamma_e \sigma^2}{\eta p_n}\left(P_{los}(\theta_n)\eta d_{nu}^{\beta_1}+\left(1-P_{los}(\theta_n)\right) d_{nu}^{\beta_2}\right)\right)\right\}\leq\zeta \label{constr_pso_org}\\
 &\ \sum_n p_n\leq P_T.
\end{align}
\end{subequations}
We define $b_n\triangleq \frac{\sigma^2}{\eta}\left(P_{los}(\theta_n)\eta d_{nu}^{\beta_1}+\left(1-P_{los}(\theta_n)\right) d_{nu}^{\beta_2}\right)$ for brevity.
Note that $\exp\left(-\frac{\gamma_e b_n}{p_n}\right)$ represents the secrecy outage probability at the $n$-th hop and is usually much smaller than 1 (since $p_{so}$ is constrained by a small $\zeta$).
Since $\prod(1-x_k)\approx 1-\sum x_k$ when $0<x_k\ll 1$ \cite{linear},
(\ref{constr_pso_org}) can be rewritten as
\begin{equation}
\sum_n\exp\left(-\frac{\gamma_e b_n}{p_n}\right)\leq\zeta
\end{equation}
Due to  $R_t=R_s+R_e$, under fixed value of $R_s$, finding the optimal value of $R_t$ is equivalent to optimizing $R_e$ (and therefore $\gamma_e$).
Based on $\gamma_e=2^{R_t-R_s}-1$ and $\gamma_c=2^{R_t}-1$, we can derive $\gamma_c=(\gamma_e+1)2^{R_s}-1$.
Setting $t_n\triangleq \gamma_e/p_n$, we can get the following optimization problem:
\begin{subequations}
\begin{align}
\min_{\gamma_e, t_n}&\ \sum_n{t_n}\frac{(\gamma_e+1)2^{R_s}-1}{\gamma_e}\\
\mathrm{s.t.}&\ \sum_n\exp\left(- b_n t_n\right)\leq\zeta\\
&\ \sum_n \frac{\gamma_e}{t_n}\leq P_T.
\end{align}
\end{subequations}

First, we focus on  the optimization of $\gamma_e$ under fixed $t_n$. This problem is equivalent to
\begin{equation} \label{sub_gammae}
\begin{split}
\min_{\gamma_e}&\ \frac{(\gamma_e+1)2^{R_s}-1}{\gamma_e}\\
\mathrm{s.t.} &\ \sum_n \frac{\gamma_e}{t_n}\leq P_T.
\end{split}
\end{equation}
It is clear that the objective function of (\ref{sub_gammae}) is decreasing with respect to $\gamma_e$. Therefore, (\ref{sub_gammae}) reaches its optimum when the inequality constraint is active. Hence, we have 
\begin{equation} \label{gamma_e_mid}
\gamma_e^*=\frac{P_T}{\sum_n 1/t_n}.
\end{equation}

Next, we aim to optimize $t_n$. 
With the approximation 
 $\gamma_c\approx\gamma_e2^{R_s}$ when $\gamma_e$ is much larger than $1$, the sub-optimization problem with respect to $t_n$ can therefore be written as 
 \begin{equation}
 \begin{split}
 \min_{t_n}&\ \sum{t_n}\\
 \mathrm{s.t.}&\ \sum_{t_n} \exp{(-b_n t_n)}\leq \zeta.
 \end{split}
 \end{equation}
This is a convex optimization problem and its globally optimal solution can be derived. Since the objective function is non-decreasing while the left-hand side of the constraint is non-increasing, this problem reaches its optimum when the inequality constraint is active. Using Lagrangian multiplier method, we can formulate the following function:
\begin{equation}
G=\sum_n t_n+\lambda\left\{\sum_n\exp\left(- b_n t_n\right)-\zeta\right\},
\end{equation}
where $\lambda>0$ represents the Lagrange multiplier.
Letting 
\begin{equation}
\frac{\partial{G}}{\partial{t_n}}=1-\lambda\exp(-\gamma_e b_n t_n)\gamma_e b_n=0,
\end{equation}
and considering the active constraint, we can finally derive 
\begin{align}
t_n&=\frac{1}{b_n}\ln\left(\sum_k\frac{b_n}{b_k\zeta}\right), \mathrm{and} \label{t_n_mid}\\
\lambda&=\frac{1}{\zeta}\sum_n b_n.
\end{align}
Then substituting (\ref{t_n_mid}) into (\ref{gamma_e_mid}) and using $t_n= \gamma_e/p_n$ and $\gamma_e=2^{R_t-R_s}-1$,
the expressions of the optimal
 $\gamma_e^*,  p_n^*$ and $R_t^*$ can be expressed as
\begin{align}
\gamma_e^*&=\frac{P_T}{\sum_n\frac{b_n}{\ln \left(b_n \sum_k \frac{1}{\zeta b_n}\right)}} \label{opt_gammae},\\
p_n^*&=\frac{\gamma_e^* b_n}{\ln\left(\sum_k\frac{b_n}{b_k\zeta}\right)},\label{opt_pn_secrecy}\ \mathrm{and}\\
R_t^*&=\log_2(\gamma_e^*+1)+R_s.\label{opt_rt_secrecy}
\end{align}

From (\ref{opt_gammae}), it is clear that $\gamma_e$ is a decreasing function of $\zeta$, which indicates that as the secrecy constraint getting relaxed (i.e. $\zeta$ increasing), less cost is required to confront the eavesdropping.


\subsection{Secrecy Rate Optimization}
Having obtained $p_n^*$ and $R_t^*$, now we deal with the optimization of $R_s$ under any given $N$.
The optimal $R_s^*$ that maximize $\Phi_S$ can be expressed as
\begin{equation}\label{subprob_rate_secrecy}
 \max_{R_s}\ \exp\left(-\frac{(\gamma_e^*+1)2^{R_s}-1}{\gamma_e^*}A_1\right)R_s.
\end{equation}
in which
\begin{equation}
A_1\triangleq\sum_n\frac{d_{tr}^\alpha \sigma^2 }{b_n} \ln\sum_k \frac{b_n}{b_k\zeta}.
\end{equation}

\begin{proposition}\label{property_quasi_concave}
Problem (\ref{subprob_rate_secrecy}) is quasi-concave.
\end{proposition}
\begin{proof}
First, we derive the first derivative of the objective function with respect to $R_s$, which can be written as
\begin{align}
\frac{\mathrm{d}}{\mathrm{d}R_s}\exp\left(-\frac{(\gamma_e^*+1)2^{R_s}-1}{\gamma_e^*}A_1\right)R_s
=\exp\left(-\frac{(\gamma_e^*+1)2^{R_s}-1}{\gamma_e^*}A_1\right)
\left(1- \frac{A_1\ln2}{\gamma_e^*}(\gamma_e^*+1) 2^{R_s}R_s\right). \nonumber
\end{align}
Let the first derivative equal to zero, and we obtain
\begin{align} \label{first_deri}
1- \frac{A_1}{\gamma_e^*}(\gamma_e^*+1)\ln2\cdot 2^{R_s}R_s=0,
\end{align}
which can be transformed as
\begin{equation}
\frac{\gamma_e^*}{A_1(\gamma_e^*+1)}=e^{\ln2\cdot R_s}\ln2\cdot R_s. \nonumber
\end{equation}
Hence, the solution of (\ref{first_deri}) is
\begin{equation}\label{rs_quasiconvex}
R_s^*=\frac{1}{\ln2}W_0\left(\frac{\gamma_e^*}{A_1(\gamma_e^*+1)}\right),
\end{equation}
where $W_0(\cdot)$ represents the principal branch of the Lambert W function.

The second derivative of the objective function with respect to $R_s$ can be expressed as
\begin{align}
&\frac{\mathrm{d^2}}{\mathrm{d}{R_s}^2}\exp\left(-\frac{(\gamma_e^*+1)2^{R_s}-1}{\gamma_e^*}A_1\right)R_s \nonumber\\
=&\exp\left(-\frac{(\gamma_e^*+1)2^{R_s}-1}{\gamma_e^*}A_1\right)
\left(-\frac{A_1(\gamma_e^*+1)}{\gamma_e^*}\ln2\cdot2^{R_s}\right)
\left(2+\ln2\left(1-\frac{A_1}{\gamma_e^*}(\gamma_e^*+1)2^{R_s}\right)R_s\right)\label{second_order},
\end{align}
When the first derivative equals to zero, due to (\ref{first_deri}), 
we have
\begin{equation}
2+\ln2\left(1-\frac{A_1}{\gamma_e^*}(\gamma_e^*+1)2^{R_s}\right)R_s=1+\ln2\cdot R_s>0,
\end{equation}
and (\ref{second_order}) is less than 0. 
Therefore,  (\ref{subprob_rate_secrecy}) is quasi-concave in $R_s$.$\hfill\blacksquare$
\end{proof}

Based on Proposition \ref{property_quasi_concave},
problem (\ref{subprob_rate_secrecy}) reaches its optimum when (\ref{first_deri}) holds, and the optimal secrecy rate $R_s^*$ can be calculated via (\ref{rs_quasiconvex}).

\subsection{Number of Hops Optimization}
Finally, the number of hops $N$ remains to be optimized.
Having obtained the expressions of the optimal $p_n^*$, $R_t^*$ and $R_s^*$ through (\ref{opt_pn_secrecy}), (\ref{opt_rt_secrecy}) and (\ref{rs_quasiconvex}), the throughput optimization can be rewritten as
\begin{equation}\label{subprob_n_secure}
\max_N\ \Phi_S=\frac{1}{N}\exp\left\{-\frac{A_1}{\gamma_e^*} \left[(\gamma_e^*+1)2^{\frac{1}{\ln2}W_0\left(\frac{\gamma_e^*}{A_1(\gamma_e^*+1)}\right)}-1\right]\right\}
\frac{1}{\ln2}W_0\left(\frac{\gamma_e^*}{A_1(\gamma_e^*+1)}\right),
\end{equation}
%
where $\gamma_e^*$ and $A_1$ are both functions of $N$.
The expression of $\Phi_S$ is quite complicated w.r.t $N$,
which makes it hard to get a specific closed-form expression for the optimal number of hops $N^*$. Nevertheless, $N$ can be derived via a one-dimensional searching in the set of positive integers, which is efficient for practice networks. The numerical results of $N^*$ are shown in Section V.

Till now, we have finished the optimization of security transmission problem. For any given number of hops $N$, the secrecy rate,
transmission rate, and transmit power can be optimized according to the closed-form expressions (\ref{rs_quasiconvex}), (\ref{opt_rt_secrecy}), and (\ref{opt_pn_secrecy})
sequentially, and $(N+2)$ values are calculated.
Supposing the optimal number of hops is searched from $1$ to $\tilde{N}$, the number of values to be computed in the entire procedure is in the order of $\mathcal{O}(\tilde{N}^2)$ .

\section{Covert Communication Optimization}
In this section, we discuss the solution of the 
covert communication optimization problem (\ref{detect_problm}).
Since $P_c$ is the function of $R_t$ and $p_n$, similar to (\ref{secure_steps}), (\ref{detect_problm}) can be solved by optimizing $N$, $R_t$, and $p_n$ sequentially, i.e.,
\begin{equation}
\max_{N,R_t,p_n}\frac{1}{N}P_c R_t \Leftrightarrow \max_N \frac{1}{N}\left(\max_{R_t}R_t\left(\max_{p_n}P_c\right)\right),
\end{equation}
The procedure to solve problem (\ref{detect_problm}) can be separated into three steps:
First, under any given $N$ and $R_t$, the transmit power $p_n$ is optimized; Then we derive the optimal $R_t$ for a given number of hops. Finally, the optimal $N$ is obtained to maximize the throughput.

\subsection{Transmit Power Optimization}

We first derive the expression of the relative entropy $D(Q_1^{(n)}||Q_0^{(n)})$. 
As previously stated, the possibility for the $n$-th air-to-ground channel being line-of-sight is $P_{los}(\theta_n)$ while for the path being NLoS is $P_{nlos}(\theta_n)$. Therefore, for the $n$-th hop, we have
\begin{align*}
D(Q_1^{(n)}||Q_0^{(n)})&=D\left[(P_{los}(\theta_n)Q_{1,los}^{(n)}+P_{nlos}(\theta_n)Q_{1,nlos}^{(n)})||Q_0^{(n)}\right],
\end{align*}
where 
$Q_{1,los}^{(n)}$ ($Q_{1,nlos}^{(n)}$) denote the joint probability distribution of the received signals on the $n$-th hop when the air-to-ground channel is LoS (NLoS) and $H_1$ is true.
Furthermore, due to the convexity of $F(t)\triangleq t\ln\frac{t}{a}$, we can derive
\begin{align}
&\ \quad D(Q_1^{(n)}||Q_0^{(n)})\nonumber\\
&=\int \left[P_{los}(\theta_n)Q_{1,los}^{(n)}(x)+P_{nlos}(\theta_n)Q_{1,nlos}^{(n)}(x)\right]\ln\frac{P_{los}(\theta_n)Q_{1,los}^{(n)}(x)+P_{nlos}(\theta_n)Q_{1,nlos}^{(n)}(x)}{Q_0(x)}dx \nonumber\\
&\leq \int \left[P_{los}Q_{1,los}^{(n)}(x)\ln\frac{Q_{1,los}^{(n)}(x)}{Q_0(x)} + P_{nlos}Q_{1,nlos}^{(n)}(x)\ln\frac{Q_{1,nlos}^{(n)}(x)}{Q_0(x)}\right]dx \label{covert_jensen1}\\
&= P_{los}(\theta_n)D(Q_{1,los}^{(n)}||Q_0^{(n)})+P_{nlos}(\theta_n)D(Q_{1,nlos}^{(n)}||Q_0^{(n)}) \nonumber\\
&=L P_{los}(\theta_n)D(q_{1,los}^{(n)}||q_0^{(n)})+L P_{nlos}(\theta_n)D(q_{1,nlos}^{(n)}||q_0^{(n)}).\label{covert_jensen2}
\end{align}
 where $q_{1,los}^{(n)}$, $q_{1,nlos}^{(n)}$ and $q_0^{(n)}$ represent the probability distribution of one received signal at the UAV on the $n$-th hop when the air-to-ground channel are LoS, NLoS and when no transmission happens, respectively. 
(\ref{covert_jensen1}) holds due to the Jensen's inequality, and (\ref{covert_jensen2}) holds for the independence of received signals that are transmitted through the same hop. 


Now we can derive the specific expression of the relative entropy.
For two complex Gaussian distributions $p= \mathcal{CN}(\mu_p,\sigma^2_p)$ and $q=\mathcal{CN}(\mu_q,\sigma^2_q)$, 
$D(p||q)$ can be calculated via \cite{covert_multihop}

\begin{equation}\label{entropy_calculate}
D(p||q)=\ln\frac{\sigma^2_q}{\sigma^2_p}+
\frac{\sigma^2_p}{\sigma^2_q}+\frac{(\mu_q-\mu_p)^2}{\sigma^2_q}-1.
\end{equation}
Therefore, we can derive 
\begin{align*}
D(q_{1,los}^{(n)}||q_0^{(n)})&=\frac{p_n}{d_{nu}^{\beta_1}\sigma^2}-\ln\left(1+\frac{p_n}{d_{nu}^{\beta_1}\sigma^2}\right)\quad \mathrm{and}\\
D(q_{1,nlos}^{(n)}||q_0^{(n)})&=\frac{\eta p_n}{d_{nu}^{\beta_2}\sigma^2}-\ln\left(1+\frac{\eta p_n}{d_{nu}^{\beta_2}\sigma^2}\right).
\end{align*}
Using $\ln(1+x)\geq x-x^2/2$ when $x\geq0$, we can get
\begin{align*}
D(q_{1,los}^{(n)}||q_0^{(n)})&\leq \frac{1}{2}\left(\frac{p_n}{d_{nu}^{\beta_1}\sigma^2}\right)^2\quad \mathrm{and}\\
D(q_{1,nlos}^{(n)}||q_0^{(n)})&\leq\frac{1}{2}\left(\frac{\eta p_n}{d_{nu}^{\beta_2}\sigma^2}\right)^2.
\end{align*}
As a result, 
\begin{align}
D(Q_1||Q_0)&=\sum_n D(Q_1^{(n)}||Q_0^{(n)})\\
&\leq \sum_n\left[ P_{los}(\theta_n)D(Q_{1,los}^{(n)}||Q_0^{(n)})+P_{nlos}(\theta_n)D(Q_{1,nlos}^{(n)}||Q_0^{(n)})\right]\\
&\leq \frac{L}{2} \sum_n \left[P_{los}(\theta_n)\left(\frac{p_n}{d_{nu}^{\beta_1}\sigma^2}\right)^2+P_{nlos}(\theta_n) \left(\frac{\eta p_n}{d_{nu}^{\beta_2}\sigma^2}\right)^2\right],\label{entropy_express}
\end{align}
The constraint can be rewritten as
\begin{equation}
\frac{L}{2} \sum_n \left[P_{los}(\theta_n)\left(\frac{p_n}{d_{nu}^{\beta_1}\sigma^2}\right)^2+P_{nlos}(\theta_n) \left(\frac{\eta p_n}{d_{nu}^{\beta_2}\sigma^2}\right)^2\right]\leq 2\epsilon^2,
\end{equation}
which is much stricter than (\ref{entropy_constraint}).


Now we tackle the transmit power optimization problem under any given $N$ and $R_t$. The aim of this subproblem is to maximize the connection probability $P_c$, which can be rewritten as
\begin{equation}
\begin{split}
\max_{p_n}\ & \exp\left(-\sum_n\frac{\gamma_c d^\alpha_{tr}\sigma^2}{p_n}\right)\\
\mathrm{s.t.}\ &\frac{L}{2} \sum_n \left[P_{los}(\theta_n)\left(\frac{p_n}{d_{nu}^{\beta_1}\sigma^2}\right)^2+P_{nlos}(\theta_n) \left(\frac{\eta p_n}{d_{nu}^{\beta_2}\sigma^2}\right)^2\right]\leq 2\epsilon^2,\\
 &\sum_n p_n\leq P_T,
\end{split}
\end{equation}
and is equivalent to
\begin{equation}\label{subprob_power_poly}
\begin{split}
\min_{p_n}&\ \sum_n\frac{1}{p_n}\\
\mathrm{s.t.}&\ \frac{L}{2} \sum_n \left[P_{los}(\theta_n)\left(\frac{p_n}{d_{nu}^{\beta_1}\sigma^2}\right)^2+P_{nlos}(\theta_n) \left(\frac{\eta p_n}{d_{nu}^{\beta_2}\sigma^2}\right)^2\right]\leq 2\epsilon^2,\\
&\ \sum_n p_n\leq P_T.
\end{split}
\end{equation}
This is a convex problem with respect to $p_n$ and its globally optimal solution can be obtained numerically.
From (\ref{subprob_power_poly}) we can discover that with a larger value of $\epsilon$, the nodes would transmit signals with larger power, which would increase the possibility that the UAV makes correct decision.

\subsection{Transmission Rate and Number of Hops Optimization}
With the optimal transmission power $\tilde{p}_n^*$, now we solve the
optimization problem of the transmission rate $R_t$ and the number of hops $N$.

The sub-problem toward the optimization of $R_t$ for a fixed $N$ can be expressed as
\begin{equation}
 \max_{R_t}\ \exp\left(-\left(2^{R_t}-1\right)A_2\right)R_t,
\end{equation}
where
$A_2\triangleq \sum_n\frac{d^\alpha_{tr}\sigma^2}{\tilde{p}_n^*}$.
This is a quasi-concave problem and reaches its optimum when its first derivative equals to zero. The proof is similar to that of Proposition 1, and therefore is omitted here. The optimal transmission rate $\tilde{R}_t^*$ can be expressed as
\begin{equation}\label{opt_rate}
\tilde{R}_t^*=\frac{1}{\ln2}W_0\left(\frac{1}{A_2}\right).
\end{equation}

Having obtained the expressions of $\tilde{p}_n^*$ and $\tilde{R}_t^*$, the throughput optimization problem can be rewritten as
\begin{equation}
 \max_{N}\ \frac{1}{N}\exp\left(-\left(2^{\left(\frac{1}{\ln2}W_0\left(\frac{1}{A_2}\right)\right)}-1\right)  A_2 \right)\frac{1}{\ln2}W_0\left(\frac{1}{A_2}\right).
\end{equation}
Though this problem is too complicated to derive a specific closed-form solution, the optimal number of hops can be obtained via a one-dimensional searching in the set of positive integers, which is efficient for practice networks.

\section{Numerical Results}

In this section, we present the numerical results to illustrate the performance of the proposed two methods toward the secrecy and covert communication problems, respectively.
We suppose that the distance between the source and the destination nodes is $D=300$m. The UAV is located at a height of $h$, exactly above the midpoint between the source and the destination nodes, and the path loss exponents of the terrestrial and air-to-ground channels are $\alpha=3$, $\beta_1=2.5$, and $\beta_2=2.8$, respectively \cite{height_pathloss}.
Considering an urban environment, the parameters $B$ and $C$ are set as $B=0.136$ and $C=11.95$ under $f=2,000$MHz with the excessive path-loss coefficient $\eta = -20$dB \cite{channel_angel1, opt_lap_altitude}.
The working bandwidth of the system is $10$MHz, with the noise power $\sigma_0^2=-110$dBm and the path loss at reference distance $d_0=1$m is set as $\lambda_0=-40$dB. Thus, we can obtain $\sigma^2=-70$dBm.

\subsection{Performance of the Secrecy Communication Optimization Scheme}

\begin{figure}
\centering
\includegraphics[width=3.5in]{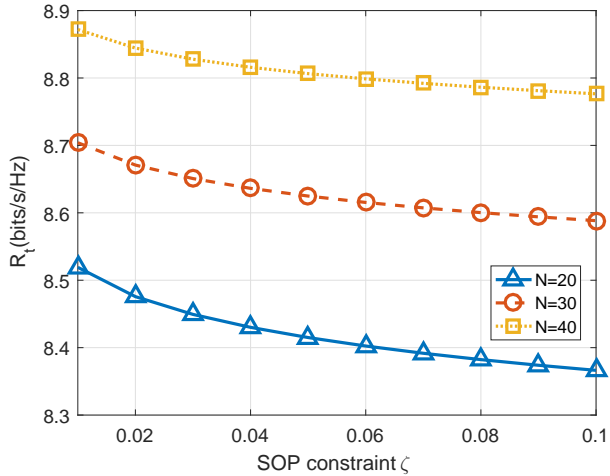}
\caption{The transmission rate $R_t$ versus the secrecy outage probability constraint $\zeta$ under fixed number of hops $N$ and $h=300$m. }
\label{fig_fixN_rt}
\end{figure}

\begin{figure}
\centering
\includegraphics[width=3.5in]{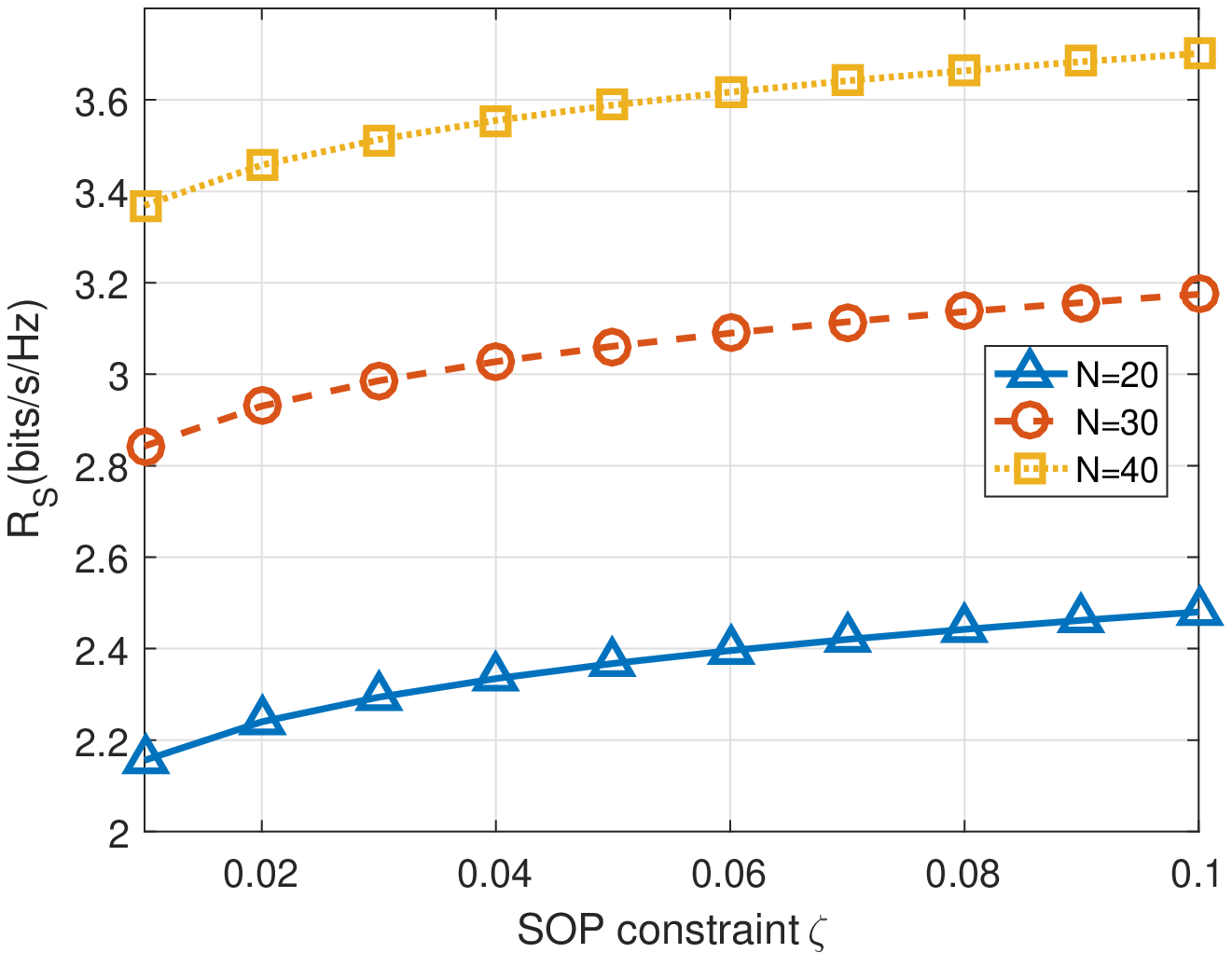}
\caption{The secrecy rate $R_s$ versus the secrecy outage probability constraint $\zeta$ under fixed number of hops $N$ and $h=300$m.}
\label{fig_fixN_rs}
\end{figure}

\begin{figure}
\centering
\includegraphics[width=3.5in]{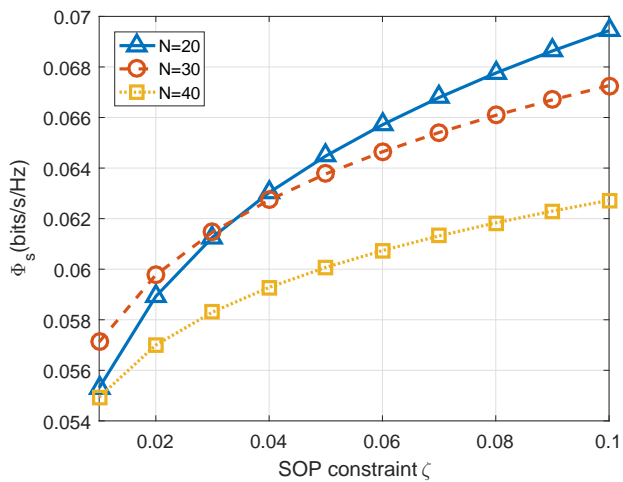}
\caption{The secrecy throughput $\Phi_S$ versus the secrecy outage probability constraint $\zeta$ under fixed number of hops $N$ and $h=300$m.}
\label{fig_fixN_throughput}
\end{figure}

\begin{figure}
	\centering
	\includegraphics[width=3.5in]{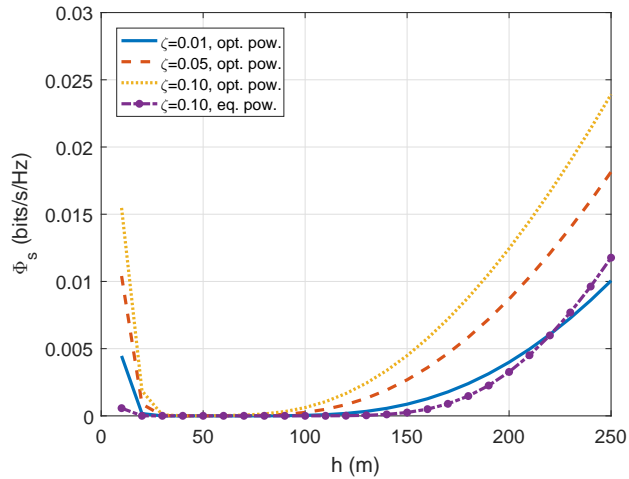}
	\caption{The secrecy throughput $\Phi_S$ versus the height of UAV $h$ when $N=7$.}
	\label{fig_throughput_height}
\end{figure}

\begin{figure}
\centering
\includegraphics[width=3.5in]{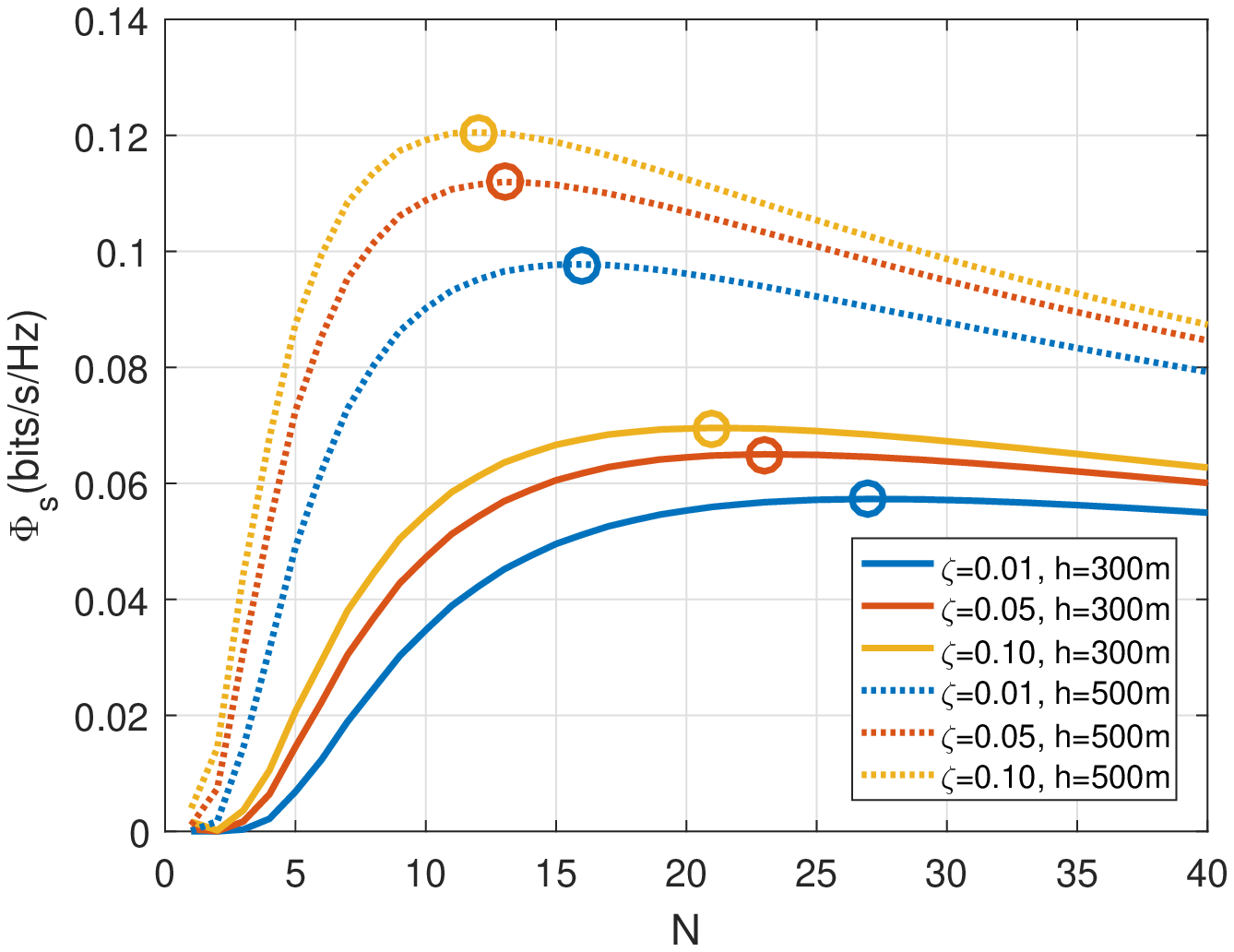}
\caption{The secrecy throughput $\Phi_S$ versus the number of hops $N$ under fixed secrecy outage probability constraint $\zeta$. The maximal values of $\Phi_S$ are marked using circles.}
\label{fig_N_throughput}
\end{figure}

In this part, we discuss the performance of the secrecy communication optimization scheme. 
The influence of the system parameters toward the optimal transmission rate, secrecy rate and the secrecy throughput is first investigated. Fig. \ref{fig_fixN_rt} and  \ref{fig_fixN_rs} depict the optimized transmission rate and  secrecy rate versus the variation of the SOP threshold $\zeta$ under different numbers of hops $N$.
As $\zeta$ increases, the transmission rate $R_t$ decreases while the secrecy rate $R_s$ increases.
The relaxation of the secrecy outage requirement enables the legitimate nodes to decrease $R_e$, which reflects the cost toward securing the confidential transmission against eavesdropping. Meanwhile, the increase of $R_s$ indicates that more information can be transmitted securely.
Besides,
as the number of hops $N$ increases, the distance between two adjacent nodes becomes shorter, leading to a decrease of required transmit power for a successful connection. In such cases, secrecy outage is less likely to happen, and both $R_t$ and $R_s$ can be enhanced so as to transmit more information.

The throughput $\Phi_S$ versus $\zeta$ under different values of $N$ are calculated according to (\ref{def_throughput})
 and the results are
 plotted in Fig. \ref{fig_fixN_throughput}.
It is clear that the throughput increases with the secrecy outage constraint being relaxed.
Besides, among the three choices of $N$, the network throughput does not always reaches the maximum at the biggest value of $N$.
Specifically, when $\zeta>0.03$, the maximal throughput is achieved as $N=30$.
This indicates the existence of an optimal value of $N$ to provide the optimal trade-off.

Fig. \ref{fig_throughput_height} reflects the impact of UAV's height on the secrecy throughput with different $\zeta$ when the number of hops is fixed as $N=7$. 
With the increase of UAV's height, the throughput first decreases and then increases gradually. This is because when the UAV has a low altitude, as $h$ increases, the throughput grows lower due to the better qualities of wiretapping channels. As $h$ keeps increasing, however, the influence of path loss becomes prominent. Hence the qualities of received signals at the UAV get worse, and more bits can be transmitted.

For comparison, we also consider a benchmark scheme (denoted as eq. pow.) in which all terrestrial transmitter have equal power, and the transmit power, $R_t$, $R_s$ and $N$ are optimized following a procedure similar to the proposed method in Section III (denoted as opt. pow.). As shown in Fig. \ref{fig_throughput_height},
for secrecy communications, our proposed scheme with different transmit power outperforms that with equal power for each node.

Then we investigate the relation between the number of hops $N$ and the throughput $\Phi_S$ and aim to find the optimal number of hops. As plotted in Fig. \ref{fig_N_throughput}, as $N$ increases, $\Phi_S$ first increases and then decreases. The maximal value of the throughput is marked with a circle on each curve in the figure, and the corresponding $N$ is exactly the optimal number of hops.
For example, when $\zeta=0.1, h=300$, the maximal throughput is obtained when there exist $21$ hops in the linear relaying network.

\subsection{Performance of the Covert Communication Optimization Scheme}

\begin{figure}
\centering
\includegraphics[width=3.5in]{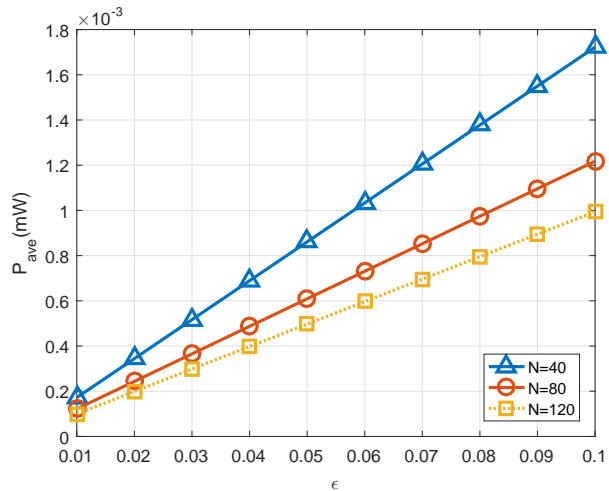}
\caption{The average transmit power $P_{ave}$ versus the detection probability constraint $\kappa$ under fixed number of hops $N$ and $h=300$m.}
\label{fig_fixN_power_covert}
\end{figure}

\begin{figure}
\centering
\includegraphics[width=3.5in]{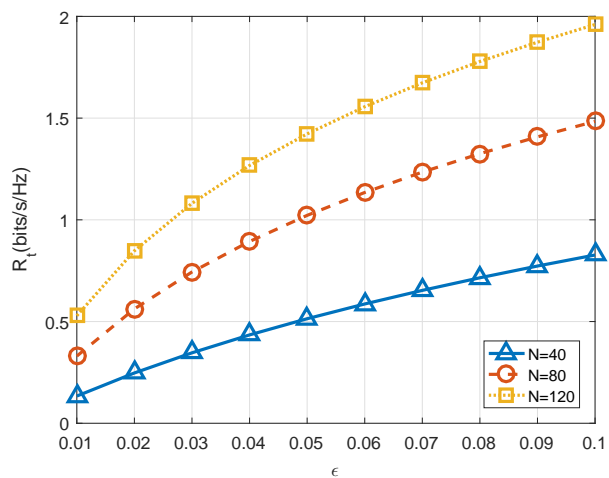}
\caption{The transmission rate $R_t$ versus the detection probability constraint $\kappa$ under fixed number of hops $N$ and $h=300$m.}
\label{fig_fixN_rate_covert}
\end{figure}

\begin{figure}
\centering
\includegraphics[width=3.5in]{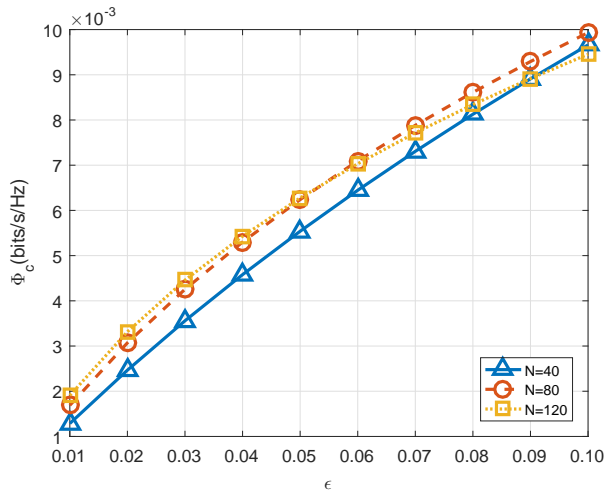}
\caption{The transmit throughput $\Phi_C$ versus the detection probability constraint $\kappa$ under fixed number of hops $N$ and $h=300$m.}
\label{fig_fixN_throughput_covert}
\end{figure}

\begin{figure}
	\centering
	\includegraphics[width=3.5in]{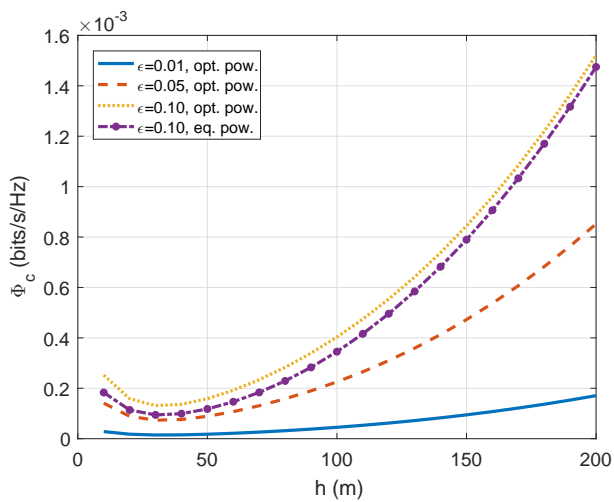}
	\caption{The transmit throughput $\Phi_c$ versus the height of UAV $h$ when $N=3$.}
	\label{fig_covert_throughput_height}
\end{figure}

\begin{figure}
\centering
\includegraphics[width=3.5in]{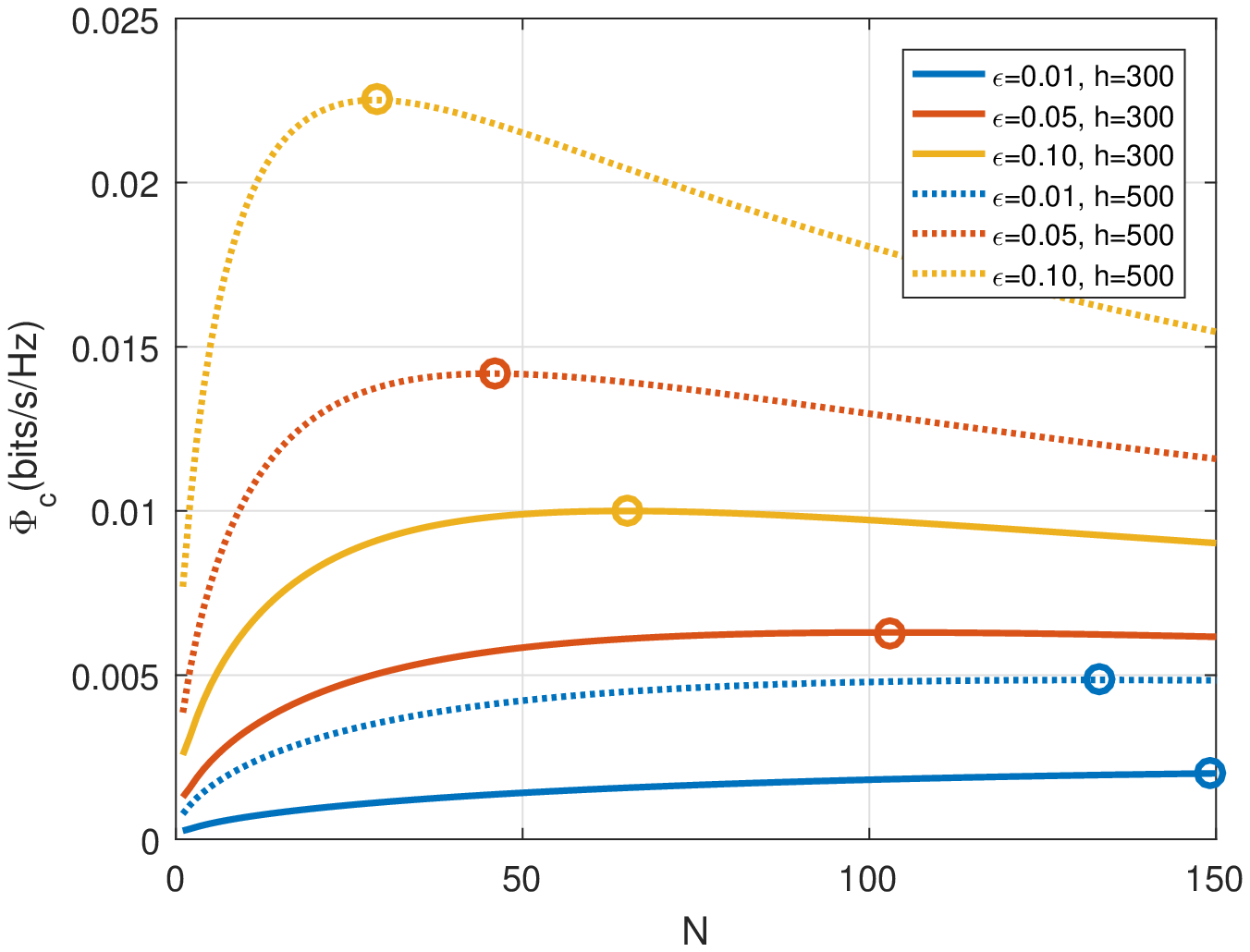}
\caption{The transmit throughput $\Phi_C$ versus the number of hops $N$ under fixed detection probability constraint $\kappa$. The maximal values of $\Phi_C$ are marked using circles.}
\label{fig_N_throughput_covert}
\end{figure}

The performance of the covert communication optimization scheme is discussed in this part.
We assume that the length of the code word $L=10$.
The influence of the system parameters toward optimized variables is studied.
The optimal transmit power, transmission rate and throughput versus  $\epsilon$ over different values of $N$ are plotted in
Fig. \ref{fig_fixN_power_covert},
Fig. \ref{fig_fixN_rate_covert}, and
Fig. \ref{fig_fixN_throughput_covert}, respectively.
The average transmit power, defined as $P_\mathrm{ave}\triangleq\frac{1}{N}\sum_{n}p_n$, increases as $\epsilon$ becomes larger while decreases as $N$ increases.
The increase of $\epsilon$ indicates a looser constraint on the detection error rate, which allows the transmitters to enhancing their transmit power without violating the covert transmission constraint.
As $N$ increases, the distance for one hop becomes shorter, and less power is required to transmit the signal.
In Fig. \ref{fig_fixN_rate_covert}, as  $\epsilon$ increases, $R_t$ becomes higher accordingly. 
The relaxation of covert constraint enables the increase of transmit power, which leads to the enhancement of received SNR at each terrestrial receiver and a higher channel capacity. Hence more bits can be transmitted.



The throughput $\Phi_C$
is calculated
and the results are depicted in Fig. \ref{fig_fixN_throughput_covert}.
It is clear that the throughput performance does not always becomes better with the increase of the number of hops. In particular, when $\epsilon$ is higher than $0.05$, $\Phi_C$ reaches the maximum when there are $80$ hops among the three provided values of $N$,
while increasing $N$ to $120$ could not lead to a better performance. This indicates the existence of an optimal number of hops.

The impact of UAV's height on the transmit throughput with different $\zeta$ and fixed $N$ can be reflected from Fig. \ref{fig_covert_throughput_height}. 
As $h$ increases from a very low altitude, the throughput first becomes smaller due to the better qualities of wiretapping channels. Then as $h$ keeps increasing, path loss becomes the dominant impact that worsens the qualities of UAV's received signals. Hence the throughput gradually increases.

In Fig. \ref{fig_covert_throughput_height}, we also compare the proposed scheme in Section IV (denoted as opt. pow.) with a benchmark scheme (denoted as eq. pow.) in which each terrestrial transmitter has identical power and the transmit power, $R_t$ and $N$ are optimized following a procedure similar to the proposed method.
As the numerical results show, the scheme with optimized transmit power outperforms that with equal power for each node in covert communication.

The relation between the number of hops $N$ and the throughput $\Phi_C$ is studied and the optimal numbers of hops under different values of  $\epsilon$ are derived. As depicted in Fig. \ref{fig_N_throughput_covert}, $\Phi_C$ first increases and then decreases with the increase of $N$. 
The optimal values for different values of  $\epsilon$ are marked on the curves. $\Phi_C$ reaches its optimum at $N=133,46,$ and $29$ when $h=500$m and $\kappa=0.01,0.05$ and $0.1$, respectively.

\section{Conclusion}
In this paper, we considered the problem of secure communication and covert communication for
a pair of  terrestrial legitimate users.
To work against the surveillance of the UAV, multiple relays were deployed to assist the transmission.
Considering the trade-off between the secrecy/covertness and the efficiency of the multi-hop transmission, the optimization of the transmit power, coding rates, and the number of hops was investigated to maximize the throughput.
In the secrecy problem, the expressions of the connection probability and the secrecy outage probability were derived.
The optimal transmit power and coding rates could be obtained via closed-form expressions and the optimal number of hops was also evaluated.
For covert communication problem, the UAV's detection error rate was constrained. 
The sub-problem of transmit power optimization was convex and the global optimal solution could be derived numerically. 
Numerical simulation illustrated the performance of multi-hop networks.
Since the secrecy and covert communication problems were discussed under different assumptions and had different targets, comparing the results under the two different conditions was meaningless. 


\linespread{1.2}


\begin{thebibliography}{99}


\bibitem{survey_RZHANG}Y. Zeng, R. Zhang and T. J. Lim, ``Wireless communications with unmanned aerial vehicles: opportunities and challenges,'' {\it{IEEE Commun. Mag.}}, vol. 54, no. 5, pp. 36-42, May 2016.




\bibitem{uav_trans_receive1}
L. Zhou, Z. Yang, S. Zhou and W. Zhang, ``Coverage probability analysis of UAV cellular networks in urban environments,'' {\it{2018 IEEE ICC Workshops}}, Kansas City, MO, 2018, pp. 1-6.

\bibitem{uav_trans_receive2}
X. Wang, H. Zhang and V. C. M. Leung, ``Modeling and performance analysis of UAV-assisted cellular networks in isolated regions,'' {\it{2018 IEEE ICC Workshops}}, Kansas City, MO, 2018, pp. 1-6.

\bibitem{uav_broadcast}
Q. Wu, J. Xu and R. Zhang, ``UAV-Enabled broadcast channel: trajectory design and capacity characterization,'' {\it{2018 IEEE ICC Workshops}}, Kansas City, MO, 2018, pp. 1-6.



\bibitem{uav_relay1}
S. Zeng, H. Zhang, K. Bian and L. Song, ``UAV relaying: power allocation and trajectory optimization using decode-and-forward protocol,'' {\it{2018 IEEE ICC Workshops}}, Kansas City, MO, 2018, pp. 1-6.






\bibitem{uav_relay3}
R. Fan, J. Cui, S. Jin, K. Yang and J. An, ``Optimal node placement and resource allocation for UAV relaying network,'' {\it{ IEEE Commun. Lett.}}, vol. 22, no. 4, pp. 808-811, April 2018.


\bibitem{channel_freespace1}
Y. Zeng, R. Zhang and T. J. Lim, ``Throughput maximization for UAV-enabled mobile relaying systems,'' {\it{ IEEE Trans. Commun.}}, vol. 64, no. 12, pp. 4983-4996, Dec. 2016.

\bibitem{channel_freespace2}
G. Zhang, Q. Wu, M. Cui and R. Zhang, ``Securing UAV communications via trajectory optimization,'' {\it{2017 IEEE Global Communications Conference}}, Singapore, 2017, pp. 1-6.


\bibitem{channel_small1}
S. Zhang, H. Zhang, Q. He, K. Bian and L. Song, ``Joint trajectory and power optimization for UAV relay networks,'' {\it{ IEEE Commun. Lett.}}, vol. 22, no. 1, pp. 161-164, Jan. 2018.
\bibitem{channel_small2}
J. Li and Y. Han, ``Optimal resource allocation for packet delay minimization in multi-layer UAV networks,'' {\it{ IEEE Commun. Lett.}}, vol. 21, no. 3, pp. 580-583, March 2017.

\bibitem{nakagami}
J. Baek, S. I. Han and Y. Han, ``Optimal resource allocation for non-orthogonal transmission in UAV relay systems,'' {\it{ IEEE Wireless Commun. Lett.}}, vol. 7, no. 3, pp. 356-359, June 2018.

\bibitem{rician}
 M. M. Azari, F. Rosas, K. C. Chen and S. Pollin, ``Optimal UAV positioning for terrestrial-aerial communication in presence of fading,'' {\it{2016 IEEE Global Communications Conference (GLOBECOM)}}, Washington, DC, 2016, pp. 1-7.


\bibitem{channel_angel1}
M. Mozaffari, W. Saad, M. Bennis and M. Debbah, ``Unmanned aerial vehicle with underlaid device-to-device communications: performance and tradeoffs,'' {\it{ IEEE Trans. Wireless Commun.}}, vol. 15, no. 6, pp. 3949-3963, June 2016.

\bibitem{opt_lap_altitude}
A. Al-Hourani, S. Kandeepan and S. Lardner, ''Optimal LAP Altitude for Maximum Coverage,'' {\textit{ IEEE Wireless Communications Letters}}, vol. 3, no. 6, pp. 569-572, Dec. 2014.


\bibitem{Liu_uav_active_eve}
C. Liu, T. Q. S. Quek and J. Lee, ``Secure UAV communication in the presence of active eavesdropper,'' {\it{ 2017 9th International Conference on Wireless Communications and Signal Processing (WCSP)}}, Nanjing, 2017, pp. 1-6.

\bibitem{PELE}
X. Wang, K. Li, S. S. Kanhere, D. Li, X. Zhang and E. Tovar, ``PELE: Power efficient legitimate eavesdropping via jamming in UAV communications,'' {\it{2017 13th International Wireless Communications and Mobile Computing Conference (IWCMC)}}, Valencia, 2017, pp. 402-408.

\bibitem{uav_eave}
Q. Song, S. Jin, F. Zheng, S. Zhang, ``UAV Wireless Information Surveillance via Proactive Eavesdropping,'' {\it{Proc. Int. Conf. Wireless Commun. Signal Process. (WCSP)}}, Nanjing, China, 2017, pp. 1--6.

\bibitem{hwang_phsuav}
H.-M.Wang, X. Zhang, and J.-C. Jiang,``UAV-involved wireless physical-layer secure communications: overview and research directions'', IEEE Wireless Communications, Oct. 2019, pp.2--9.

\bibitem{wyner}
A. D. Wyner, ``The wire-tap channel,'' {\it{ Bell Syst. Tech.J.}}, vol. 54, no. 8, pp. 1355-1387, Oct. 1975.


\bibitem{Yang2015Safeguading}
N. Yang, L. Wang, G. Geraci, M. Elkashlan, J. Yuan, and M. D. Renzo, ``Safeguarding 5G wireless communication networks using physical tier security,'' \emph{IEEE Commun. Mag.}, vol. 53, no. 4, pp. 20--27, Apr. 2015.




\bibitem{Wang2016Physical}
H.-M. Wang, T.-X. Zheng, J. Yuan, D. Towsley, and M. H. Lee,
``Physical layer security in heterogeneous cellular networks,'' \emph{IEEE Trans. Commun.}, vol. 64, no. 3, pp. 1204--1219, Mar. 2016.

\bibitem{Wang2016Physical_book}
H.-M. Wang and T.-X. Zheng, \emph{Physical Layer Security in Random Cellular Networks}. Singapore: Springer Press, 2016.

\bibitem{Wang2015Enhancing}
H.-M. Wang and X.-G. Xia, ``Enhancing wireless secrecy via cooperation: Signal design and optimization,'' \emph{IEEE Commun. Mag.}, vol. 53, no. 12,
pp. 47--53, Dec. 2015.

\bibitem{R1_review1}
G. Zhang, Q. Wu, M. Cui and R. Zhang, ``Securing UAV Communications via Joint Trajectory and Power Control," \emph{IEEE Trans.  Wireless Commun.}, vol. 18, no. 2, pp. 1376-1389, Feb. 2019.



\bibitem{R2_review1}
M. Cui, G. Zhang, Q. Wu and D. W. K. Ng, ``Robust Trajectory and Transmit Power Design for Secure UAV Communications," \emph{ IEEE Transactions on Vehicular Technology}, vol. 67, no. 9, pp. 9042-9046, Sept. 2018.

\bibitem{R3_review1}
H. Lu, H. Zhang, H. Dai, W. Wu and B. Wang, ``Proactive Eavesdropping in UAV-Aided Suspicious Communication Systems," \emph{IEEE Transactions on Vehicular Technology}, vol. 68, no. 2, pp. 1993-1997, Feb. 2019.


%
%







\bibitem{limit_awgn}
 B. A. Bash, D. Goeckel, and D. Towsley, ``Limits of reliable communication with low probability of detection on AWGN channels,'' {\it{ IEEE J. Sel. Areas Commun.}}, vol. 31, no. 9, pp. 1921-1930, Sep. 2013.



\bibitem{memoryless}
L. Wang, G. W. Wornell, and L. Zheng, ``Fundamental limits of communication with low probability of detection,'' {\it{ IEEE Trans. Inf. Theory}}, vol. 62, no. 6, pp. 3493-3503, Jun. 2016.


\bibitem{binary}
 P. H. Che, M. Bakshi, and S. Jaggi, ``Reliable deniable communication: Hiding messages in noise,'' {\it{  Proc. IEEE Int. Symp. Inf. Theory}}, Jul. 2013, pp. 2945-2949.



\bibitem{noise_uncertain}
B. He, S. Yan, X. Zhou, and V. K. N. Lau, ``On covert communication with noise uncertainty,'' {\it{ IEEE Commun. Lett.}}, vol. 21, no. 4, pp. 941-944, Apr. 2017.

\bibitem{channel_uncertain}
K. Shahzad, X. Zhou, and S. Yan, ``Covert communication in fading channels under channel uncertainty,'' {\it{ Proc. IEEE VTC Spring}}, Jun. 2017, pp. 1-5.



\bibitem{finite_length}
S. Yan, B. He, Y. Cong, and X. Zhou, ``Covert communication with finite blocklength in AWGN channels,'' {\it{Proc. IEEE ICC}}, May 2017, pp. 1-6.


\bibitem{covert_multihop}
A. Sheikholeslami, M. Ghaderi, D. Towsley, B. A. Bash, S. Guha and D. Goeckel, ``Multi-hop routing in covert wireless networks,'' {\it{IEEE Trans. Wireless Commun.}}, vol. 17, no. 6, pp. 3656-3669, June 2018.


\bibitem{R1_review3}
Q. Wu, W. Mei, and R. Zhang
``Safeguarding Wireless Networks with UAVs: A Physical Layer Security Perspective,’’ {\it{IEEE Wireless Communications}}, 2019, to appear.

\bibitem{adhoc1}
Y. Xu, J. Liu, O. Takahashi, N. Shiratori and X.Jiang, ``SOQR: secure optimal QoS routing in wireless ad-hoc networks,''
{\it{IEEE WCNC}}, pp. 1-6, Mar. 2017.


\bibitem{adhoc12}
H.-M. Wang, Y. Zhang, D. W. K. Ng and M. H. Lee, ``Secure routing with power optimization for Ad-hoc networks,'' {\it{ IEEE Trans.  Commun.}},  vol. 66, no. 10, pp. 4666-4679, Oct. 2018.


\bibitem{linear}
J. Yao, X. Zhou, Y. Liu and S. Feng, ``Secure transmission in linear multihop relaying networks,'' {\it{IEEE Trans. Wireless Commun.}}, vol. 17, no. 2, pp. 822-834, Feb. 2018.









  \bibitem{linear1}
M. Sikora, J. N. Laneman, M. Haenggi, D. J. Costello and T. E. Fuja, ``Bandwidth- and power-efficient routing in linear wireless networks,'' {\it{IEEE Trans. Inf. Theory}}, vol. 52, no. 6, pp. 2624-2633, Jun. 2006.

\bibitem{DFsecure}
T.-X. Zheng, H.-M. Wang, F. Liu, and M. H. Lee, ``Outage constrained secrecy throughput maximization for DF relay networks,'' \emph{IEEE Trans. Commun.}, vol. 63, no. 5, pp. 1741-1755, May 2015.



     
     
     






     
\bibitem{height_pathloss}J. Baek, S. I. Han and Y. Han, "Optimal resource allocation for non-orthogonal transmission in UAV relay systems," {\it{IEEE Wireless Communications Letters}}, vol. 7, no. 3, pp. 356-359, June 2018.











\end{thebibliography}
\end{document}